\documentclass[letterpaper, 10 pt, conference]{IEEEconf}  

\usepackage{cite,amsmath,amssymb,algorithm,algorithmic,skmath}
\usepackage{graphicx,graphics,epstopdf}
\usepackage{xcolor}
\newtheorem{thm}{Theorem}
\newtheorem{lem}[thm]{Lemma}

\newtheorem{rem}[thm]{Remark}
\newtheorem{exmp}{Example}

\title{An {eigen decomposition based} closed-form solution for the Discrete Lyapunov and Stein Equations}

\author{Aaqib Patel$^{1}$, Member, IEEE and Mohammed Zafar Ali Khan$^{2}$, Senior Member, IEEE
\thanks{The authors are with the Department of Electrical Engineering,
        Indian Institute of Technology Hyderabad, Kandi, Telangana, India 502285
        (e-mail: {\tt\small $^{1}$aaqib@iith.ac.in}; {\tt\small$^{2}$zafar@ee.iith.ac.in})}%
}

\begin{document}

\maketitle
\thispagestyle{empty}
\pagestyle{empty}

\begin{abstract}
A simple closed-form solution to the discrete Lyapunov equation
(DLE) is established {for certain families of matrices}. This solution is expressed in terms of the eigen
decomposition {(ED) for which closed-form solutions are known for all $2\times2, 3\times3$ and certain families of matrices. For general matrices, the proposed ED based closed-form solution  can be used as an efficient numerical solution when the ED can be computed.} The result is then extended to give closed-form solutions for a generalization of the DLE, called the Stein equation. The proposed explicit solution's complexity is of the same order as iterative solutions and significantly smaller than known closed-form solutions. These solutions may prove convenient for analysis and synthesis problems related to these equations due to their compact form.
\end{abstract}


\section{INTRODUCTION}\label{sec:intro}
The continuous and discrete Lyapunov matrix equations $AX-XA^T=Q$ and $X-AXA^T=Q$ are encountered in system stability analysis. Like the Bartels-Stewart algorithm \cite{2,3,4} and the Hessenberg-Schur algorithm, several numerical approaches were proposed \cite{5}. Nevertheless, closed-form solutions have their importance. In \cite{6}, solutions for the two Lyapunov equations were proposed when matrix $A$ is in companion form. For an arbitrary $A$, the integral forms \cite{7} and the matrix power form \cite{8} were proposed. Sensitivity of the Lyapunov equations was studied in \cite{9,10,11}.

A generalization of the two Lyapunov matrix equations $XF-AX=C$, called the Sylvester equation,  and $X-AXF=C$, called the Stein equation, play a vital role in the investigation of Yule-Walker type equations \cite{12,13}. Explicit solution of $XF-AX=C$, when the matrices $A$ and $F$ are both in Jordan forms, was provided in terms of finite double matrix series \cite{14}. When the matrix $F$ is in companion form, explicit solutions were given in terms of the controllability matrix and observability matrix \cite{15}. When the matrix $F$ is arbitrary, an explicit solution was established in \cite{16}. In \cite{17,18,19,20}, solutions for $XF-AX=C$ and $X-AXF=C$ are given according to the coefficients of the characteristic polynomial of matrix $A$, the so-called Faddev algorithm. Recent work on Sylvester equation and its generalizations can be found in \cite{22,23,24,25,26,27,28,38,39} and the references therein. A generalized version of the Stein matrix equations $XF-A\bar{X}=C$, $X-A\bar{X}F=C$ where $\bar{X}$ denotes the conjugate of $X$,  was investigated in \cite{29,37} and the references therein.

In this letter a we first derive a closed form solution of the DLE which is based on the eigen decomposition of $A$ and require $A$ to be non-defective. 
The results are then generalized to the Stein equation. As a side result, we also present a novel closed form solution for a particular matrix series that, to the best of our knowledge, was not known previously.



The rest of the paper is organized as follows. In Section \ref{sec:dle}, a closed-form solution of the discrete Lyapunov equation
(DLE) is  presented and in Section \ref{sec:stein}, the results are extended to the Stein equation. Numerical examples are presented in Section \ref{sec:num} and Section \ref{sec:conc} concludes the paper.
%
%
%
%
%
\section{Closed-form Solution of the DLE}\label{sec:dle}
Consider the discrete Lyapunov equation (DLE) given by
\begin{equation}\label{eq1}
  X=AXA^H+Q,
\end{equation}
where $A >0 \in \mathbb{C}^{n \times n}$ and $Q\in \mathbb{C}^{n \times n}$. Define  $\mathrm{vec}(A)$ as the stacking of the columns of the  matrix $A$ and $A \otimes B$ as the Kronecker product of the matrices $A$ and $B$. Using the property $\mathrm{vec}(ABC)=(C^T \otimes A)\mathrm{vec}(B)$, \eqref{eq1} can be written as
\begin{align}\label{eq3}
 & \mathrm{vec}(X)=(A \otimes A)\mathrm{vec}(X)+\mathrm{vec}(Q)\nonumber\\
\end{align}
Rearranging, we have the system of linear equations
\begin{align}\label{eq3a}
& \Rightarrow (I_{n^2}-(\bar{A} \otimes A))  \mathrm{vec}(X)=\mathrm{vec}(Q).
\end{align}
\begin{rem}
A closed-form solution of the discrete Lyapunov equation (DLE) can be obtained from \eqref{eq3}. However, this requires solving a system of linear equations of size $n^2$. The complexity of the best known solutions for a matrix of size $n^2$ is $O(n^{4.75})$ \cite{36}, where $O()$ denotes the order of. Practical applications would be of the order of $O(n^6)$. 
Similarly, the other known closed form solutions are of higher complexity as matrices of dimesnion $n^2$ are involved. In what follows we develop an exact solution of the DLE with  complexity $O(n^3)$.
\end{rem}
We first simplify \eqref{eq1} to a simpler problem when $A$ is non-defective. Denoting $S^{-H}= \left( S^{H}\right)^{-1}$, we have
\begin{lem}\label{lem1}
 The solution of \eqref{eq1} is obtained by solving
 \begin{equation}\label{eq4}
 Y=\Sigma Y\bar{\Sigma}+\hat{Q}
,\end{equation}
 where $A=S^{-1}\Sigma S$ is the eigen decomposition of $A$, $Y=S X S^H$ and $\hat{Q}=S Q S^H$.
\end{lem}
\begin{proof}
Let $S^{-1}\Sigma S= A$ be the eigen-decomposition of $A$, then substituting in \eqref{eq1}, we have
\begin{equation}\label{eq5}
X-S^{-1}\Sigma S X S^{H}\bar{\Sigma }S^{-H}=Q.
\end{equation}
Appropriately multiplying by $S$ and $S^{H}$, we have
\begin{eqnarray}\label{eq6}
 &&S X S^{H}-\Sigma S X S^{H}\bar{\Sigma}=S Q S^{H}. \nonumber
 \end{eqnarray}
Defining $Y=S X S^H$ and $\hat{Q}=S Q S^H$ and substituting in \eqref{eq6}, we have \eqref{eq4}.
If $Y_o$ is a solution of \eqref{eq4}, then the solution of \eqref{eq1} is given by $X_o= S^{-1} Y_o S^{-H}$ and vice-versa completing the proof.
\end{proof}
Equation \eqref{eq4} allows for closed form solution of  $X$, as is shown in the next theorem. Denoting the conjugate of $a$ as  $a^*$, we have
\begin{thm}\label{thm2}
The solution of \eqref{eq1}, denoted as $X_o$, is given by
\begin{equation}\label{eq7}
X_{o} = S^{-1} \left( SQS^H \circ M\right) S^{-H},
\end{equation}
 where $\circ$ denotes the Hadamard product, 
 {\small \begin{equation}\label{eqM}
 M=\begin{bmatrix}
              (1-|\sigma_1|^2)^{-1} & (1-\sigma_1 \sigma_2^*)^{-1} & \cdots & (1-\sigma_1\sigma_n^*)^{-1} \\
             (1-\sigma_2 \sigma_1^*)^{-1} & (1-|\sigma_2|^2)^{-1} & \cdots & (1-\sigma_2 \sigma_n^*)^{-1} \\
              \vdots & \ddots & \ddots & \vdots \\
              (1-\sigma_n \sigma_1^*)^{-1} & (1-\sigma_n \sigma_2)^{-1} & \cdots & (1-|\sigma_n|^2)^{-1} \\
            \end{bmatrix},
            \end{equation}
}
 $A=S^{-1}\Sigma S$ is  the eigen decomposition of matrix $A$ and $\sigma_i$ is the $i$-th eigenvalue  of the matrix $A$.
 \end{thm}
 \begin{proof}
Taking the $i$-th column  on both sides of  \eqref{eq4} and using matrix manipulations, we have
\vspace{-0.13in}
  \begin{eqnarray}\label{eq9}
 &&Y_i- \Sigma Y_i\sigma_i^*=\hat{Q}_i, \nonumber\\
&& \Rightarrow \left(I_n-\sigma_i^* \Sigma\right)Y_i=\hat{Q}_i\nonumber\\
&& \Rightarrow Y_i=\left(I_n-\sigma_i^* \Sigma\right)^{-1}\hat{Q}_i.
\end{eqnarray}
The diagonal matrix $\left(I-\sigma_i^* \Sigma \right)^{-1}$ can be simplified using Taylor series as follows
\begin{align*}
& \left(I-\sigma_i^* \Sigma \right)^{-1}=\hspace{3in}
 \end{align*}
\begin{align}\label{eq10}
 ~~ & \begin{bmatrix}
                                            (1-\sigma_i^*\sigma_1)^{-1} & 0& \cdots & 0 \\
                                            0 & (1-\sigma_i^*\sigma_2)^{-1} & \cdots & 0 \\
                                            \vdots & \ddots & \ddots & \vdots \\
                                            0 & \cdots & 0 & (1-\sigma_i^*\sigma_n)^{-1} \\
                                          \end{bmatrix}
\nonumber \\
                                      & = \begin{bmatrix}
                                            \sum_{l=0}^{\infty} [\sigma_i^*\sigma_1]^l & 0& \cdots & 0 \\
                                            0 & \sum_{l=0}^{\infty} [\sigma_i^*\sigma_2]^l & \cdots & 0 \\
                                            \vdots & \ddots & \ddots & \vdots \\
                                            0 & \cdots & 0 & \sum_{l=0}^{\infty} [\sigma_i^*\sigma_n]^l\\
                                          \end{bmatrix}
                                          \nonumber\\
                                          & = \sum_{l=0}^{\infty}\begin{bmatrix}
                                            [\sigma_i^*\sigma_1]^l & 0& \cdots & 0 \\
                                            0 &  [\sigma_i^*\sigma_2]^l & \cdots & 0 \\
                                            \vdots & \ddots & \ddots & \vdots \\
                                            0 & \cdots & 0 &  [\sigma_i^*\sigma_n]^l\\
                                          \end{bmatrix}
                                          \nonumber\\
                                          &
= \sum_{l=0}^{\infty} \Sigma^l \left(\sigma_i^*\right)^l.
\end{align}
Observe that we have assumed $|\sigma_i^*\sigma_j|  < 1, \forall  i,j$ while using the Taylor series expansion, which has been known previously \cite{6}.
Substituting \eqref{eq10} in \eqref{eq9}, we have
  \begin{eqnarray}\label{eq11}
Y_i&=&  \sum_{l=0}^{\infty} \Sigma^l \left(\sigma_i^*\right)^l \hat{Q}_i
.
\end{eqnarray}
Collecting the columns in a matrix, we have
\begin{align}\label{eq12}
Y&=
\left[\sum_{l=0}^{\infty} \Sigma^l \left(\sigma_1^*\right)^l \hat{Q}_1~~\sum_{l=0}^{\infty} \Sigma^l \left(\sigma_2^*\right)^l \hat{Q}_2~\cdots ~ \Sigma^l \left(\sigma_n^*\right)^l \hat{Q}_n \right],\nonumber\\
&=\sum_{l=0}^{\infty}\Sigma^l\left[\left(\sigma_1^*\right)^lSQ{S}^H_1~~\left(\sigma_2^*\right)^l SQ{S}^H_2\cdots \left(\sigma_n^*\right)^lSQ{S}^H_n \right],\nonumber\\
&=\sum_{l=0}^{\infty}\Sigma^lSQ\left[\left(\sigma_1^*\right)^l{S}^H_1~~\left(\sigma_2^*\right)^l {S}^H_2\cdots\left(\sigma_n^*\right)^l{S}^H_n \right]\nonumber\\
&
=\sum_{l=0}^{\infty} \Sigma^l SQS^H \bar{\Sigma}^l,
\end{align}
where $S_i^H,$ denotes the $i$-th column of $S^H$. Taking the entry in the $i$-row, $j$-th column of $Y$ in \eqref{eq12}, we have
\begin{align}\label{eq13}
Y_{ij}&=\sum_{l=0}^{\infty} \sigma_i^l (S_i^H)^HQS_j^H \left(\sigma_j^*\right)^l\nonumber\\&
= (S_i^H)^HQS_j^H\sum_{l=0}^{\infty} \sigma_i^l \left(\sigma_j^*\right)^l,\nonumber\\
&= (S_i^H)^HQS_j^H(1-\sigma_i \sigma_j^*)^{-1}.
\end{align}
Observe that $(S_i^H)^H$ is the $i$-th row of $S$. Using the definition of the matrix $M$ as defined in \eqref{eqM},
we have 
\begin{align}\label{eq14}
Y&= SQS^H \circ M,
\end{align}
where $A \circ B$ denotes the Hadamard product of matrices $A,B$. Using the relation between $X$ and $Y$, we have \eqref{eq7}. 
 \end{proof}
 \begin{rem}
 The solution of DLE using \eqref{eq7} requires eigen decomposition of $A$. Closed form solutions for eigenvalues and eigenvectors of all 2x2 and 3x3 matrices are well known. The eigen values and eigen vectors are also known for  many families and classes of matrices, like diagonal, Toeplitz, Hankel, etc. So, the proposed solution is a 'closed form solution' only for certain matrices. For general matrices, the proposed solution can be treated as an ED based numerical method. However, the eigen decomposition is not generally stable for large matrices and the proposed method may be used only for simplification in design and analysis, where the final numerical solutions can be obtained using other stable algorithms. Also, stability criteria like those in \cite{40} can be used to identify suitability of the proposed algorithm.
 \end{rem}
 \begin{rem}
 For computing the solution of DLE using \eqref{eq7}, eigen decomposition of $A$ is needed which is of complexity $O(n^3)$. Computing $M$ is of complexity $2n^2$ and computing the Hadamard product is of complexity $n^2$. Computing the four matrix products is of complexity $O(4n^{2.376})$, so the overall complexity is $O(n^3+4n^{2.376}+3n^2)$.  Note that using the standard matrix multiplication, the complexity will be $O (5 n^{3}+2n^2)$. Observe  that calculating $X$ using \eqref{eq7} requires the eigen values of $A$, $|\sigma_i|  < 1$ which is true for all stable systems. Further, the derivation requires the matrix $A$ to be non-defective so that the eigen decomposition exists.
 \end{rem}
 \begin{rem}
  Equation \eqref{eq12} can be used to obtain the well known series solution of the DLE; which is the standard approach. The use of Hadamard product in \eqref{eq14} allowed us to propose the new ED based closed form solution, which to the best of our knowledge was not proposed earlier.  As $Y=S X S^T$, we have from \eqref{eq12}
   \begin{align}\label{eq16}
X&= S^{-1}YS^{-H},\nonumber\\
&
=\sum_{l=0}^{\infty} S^{-1}\Sigma^l SQS^H \Sigma^l S^{-H}\nonumber\\
&
=\sum_{l=0}^{\infty} A^l  Q  \left(A^l\right)^H.
\end{align}
This allows us to give a new closed-form solution for this matrix series in Theorem \ref{thm2}.
%
 \end{rem}
 \begin{thm}\label{thm3}
 The matrix series $\sum_{l=0}^{\infty} A^l  Q  \left(A^l\right)^H$ converges  iff the absolute value of the maximum eigen value of $A$ is less than 1 so that
   \begin{align}\label{eq16a}
S^{-1} \left( SQS^H \circ M\right) S^{-H}= \sum_{l=0}^{\infty} A^l  Q  \left(A^l\right)^H.
\end{align}
 \end{thm}
 \begin{rem}
  The  result of Theorem \ref{thm2} holds for symmetric or Hermitian $A$, wherein the singular value decomposition (SVD) of  $A=O\Sigma O^H$, where $O$ is a unitary matrix and $\Sigma$ is diagonal, can be used. 
 \end{rem}
\section{Closed-form Solution of the Stein Matrix Equation}\label{sec:stein}
Consider the Stein matrix equation \cite{19}
\begin{equation}\label{eq28}
  X -AXF= Q
\end{equation}
where $A \in \mathbb{C}^{n \times n}, Q \in \mathbb{C}^{n \times p},$ and $F \in \mathbb{C}^{p \times p}$ are the given matrices and $X \in \mathbb{C}^{n \times p}$ is the matrix to be determined. 
We have
\begin{lem}\label{lem7}
 The solution of \eqref{eq28} is obtained by solving
 \begin{equation}\label{eq29}
 Y=\Theta Y \Lambda + \tilde{Q}
,\end{equation}
 where $A=S^{-1}\Theta S$, $F=V^{-1}\Lambda V$ are the eigen-decompositions of the matrices $A$, $F$ respectively and $\hat{Q}= S Q V^{-1}$.
\end{lem}
\begin{proof}
Let $U^{-1}\Theta U= A$ be the eigen-decomposition of $A$ and $V^{-1}\Lambda V=F$ be the eigen-decomposition of $F$, then substituting in \eqref{eq28}, we have
\begin{equation}\label{eq32}
X-S^{-1}\Theta S X V^{-1}\Lambda V =\hat{Q}.
\end{equation}
Pre-multiplying by $S$, and post multiplying by $V^{-1}$, we have
\[
S X V^{-1}-\Theta S X V^{-1}\Lambda =S{Q}V^{-1}.
\]
Substituting $Y=S X V^{-1}$ and $\hat{Q}=S{Q}V^{-1}$, we have \eqref{eq29}.
 If $Y_o$ is a solution of \eqref{eq29}, then the solution of \eqref{eq28} is given by $X_o= U^{-1} Y_o V$ and vice-versa completing the proof.
\end{proof}
Similar to Theorem \ref{thm2} we have Theorem \ref{thm8}, and hence the proof has been omitted.
\begin{thm}\label{thm8}
Let $X_o$ be the solution of \eqref{eq28}, then
\begin{equation}\label{eq34}
X_{o} = U^{-1} \left( UQV^{-1} \circ N\right) V,
\end{equation}
 where the  entry in $i$-th row, $j$-th column of the matrix $N \in \mathbb{C}^{n \times p}$ is defined as
$  N_{ij}=(1-\theta_i \lambda_j)^{-1}$,  $\lambda_i$ is the $i$-th eigenvalue  of the matrix $F$ with eigen decomposition $F=V^{-1} \Lambda V$ and $\theta_j$ is the $j$-th eigenvalue  of the matrix $A = U^{-1}\Theta U$.
 \end{thm}
 \begin{rem}
 Observe that we require $|\theta_i \lambda_j| \le 1, \forall i,j$.
 \end{rem}
 \section{Numerical Examples}\label{sec:num}
 As an illustration of the use of the closed-form solutions presented in this letter, we present here the results of
its application to an example  $A,Q$ matrices for the DLE.
\begin{exmp}
Consider $$A=\left[\begin{array}{ccccc}
3 & 9 & 5 & 1  \\
1 & 2 & 3 & 8  \\
4 & 6 & 6 & 6  \\
1 & 5 & 2 & 0  \\
\end{array}\right]\mbox{ and }Q=\left[\begin{array}{ccccc}
2 & 4 & 1 & 0 \\
4 & 1 & 0 & 2 \\
1 & 0 & 3 & 0 \\
0 & 2 & 0 & 1 \\
\end{array}\right].$$ The eigen decomposition of $A$ gives $$S= \begin{bmatrix}
                                                                -0.3096 & -0.7357  &  -0.5709 & -0.6480  \\
                                                                -0.6861 &  0.5214  &  -0.2285 &  1.4524  \\
                                                                -1.0883 & -0.3410  &   0.5985 &  1.3881  \\
                                                                -0.0664 & -0.8251  & - 0.0242 &  1.2688  \\
                                                              \end{bmatrix} $$
                                                              and
$\Sigma=\begin{bmatrix}
                                                                                       14.8470& 0 & 0 &0\\
                                                                                       0 & 1.4554 & 0 & 0\\ 0& 0& -0.1620&0\\0 & 0& 0& -5.1404                                                                                           \end{bmatrix}.$ \\\noindent
 Using Theorem \ref{thm2}, we obtain
  $$M=\begin{bmatrix}
                                                       -0.0046 & -0.0485  &  0.2936 &   0.0129    \\
                                                       -0.0485 & -0.8942  &  0.8092 &   0.1179    \\
                                                        0.2936 &  0.8092  &  1.0270 &   5.9875    \\
                                                        0.0129 &  0.1179  &  5.9875 &  -0.0393    \\
                                                     \end{bmatrix}
 $$ and the closed-form solution is given by $$X_o =\begin{bmatrix}
                        -11.2596  &  4.8462  &  7.1758  &  -6.0125  \\
                          4.8462  &  1.6896  & -4.3210  & -0.3908  \\
                          7.1758  & -4.3210  & -6.2960  & - 6.6110  \\
                         -6.0125  & -0.3908  &  6.6110  & -2.4587  \\
                        \end{bmatrix}.$$
 Back substituting the solution in the DLE verifies the correctness of the solution.
 \end{exmp}
\section{Conclusion}\label{sec:conc}
In this letter we have presented a closed-form expression for the discrete Lyapunov equation and its extension, the Stein equation based on the eigen decomposition. The closed-form expressions are of the same size as the coefficient matrices and as such are compact. The proposed closed form is limited by the (non) existence of closed form expressions for ED. However the proposed solution may help in simplifying design and analysis of control systems while the final solutions can be obtained by more numerically stable algorithms. Further, the computational complexity of the exact solution (when it exists, is cubic in the size of the coefficient matrices. This makes the complexity of the proposed exact solution of the same order as that of iterative solutions. The proposed approach has the potential to be applied to other matrix equations. The solutions of the dual continuous Lyapunov equations can be similarly obtained by following the approach of this paper or by using transformations as listed in literature. Perturbation analysis and development of numerically stable algorithms based on the above approach will be considered in future work.


\section*{Acknowledgment}
This work has been supported by Visvesvarya YFR fellowship to Prof. Mohammed Zafar Ali Khan and DST-INSPIRE YF fellowship to Dr. Aaqib Patel.  




%

%
%

\bibliography{IEEEabrv,autosam1}
\bibliographystyle{IEEEtran}

\end{document}